\documentclass[reqno,10pt]{article}
\usepackage{amssymb,amsthm,amsmath,amsfonts}
\usepackage{epsfig}

\theoremstyle{plain}
\begingroup
 
\newtheorem{thm}{Theorem}[section] 

\newtheorem{lemma}[thm]{Lemma}

\endgroup
\theoremstyle{definition}
\newtheorem{defn}{Definition}[section]
\newtheorem{rem}[defn]{Remark}

\theoremstyle{remark}

\numberwithin{equation}{section}
\numberwithin{figure}{section}

\begin{document}

\title{Perturbation theorems for Hele-Shaw flows and their applications}
\author{Yu-Lin Lin\textsuperscript{1}}

\date{\today}

\maketitle

\begin{abstract}
In this work, we give a perturbation theorem for strong polynomial solutions to the zero surface tension Hele-Shaw equation driven by injection or suction, so called the Polubarinova-Galin equation. This theorem enables us to explore properties of solutions with initial functions close to but are not polynomial. Applications of this theorem are given in the suction or injection case. In the former case, we show that if the initial domain is close to a disk, most of fluid will be sucked before the strong solution blows up. In the later case, we obtain precise large-time rescaling behaviors for large data to Hele-Shaw flows in terms of invariant Richardson complex moments. This rescaling behavior result generalizes a recent result regarding large-time rescaling behavior for small data in terms of moments. As a byproduct of a theorem in this paper, a short proof of existence and uniqueness of strong solutions to the Polubarinova-Galin equation is given. 
\end{abstract}
\noindent
Keywords:  Hele-Shaw flows, starlike function, rescaling behavior.

\footnotetext[1]
{Institute of Mathematics, Academia Sinica, Nankang, Taipei, Taiwan 11529 R. O. C.\\
Email: \tt{yulin@math.sinica.edu.tw }}
\section{Introduction}
This paper deals with classical zero surface tension (ZST) Hele-Shaw flows. The driving mechanism, injection or suction with a constant rate $2\pi$ or $-2\pi$ at the origin, produces a family of domains $\{\Omega(t)\}_{t\geq0}$. In two dimensions, Galin and Polubarinova-Kochina reformulated the planar model of Hele-Shaw flows by describing the domains $\{\Omega(t)\}$ by a family of conformal mappings $\{f(\xi,t)\}$ where $f(\xi,t): D\rightarrow \Omega(t)$ and $f(0,t)=0, f^{'}(0,t)>0$. Here we set 
\[f_{t}(\xi,t)=\frac{\partial}{\partial t}f(\xi,t),\quad f^{'}(\xi,t)=\frac{\partial}{\partial\xi}f(\xi,t), \quad D=D_{1}(0),\quad  D_{r}=D_{r}(0)\] 
where $D_{r}(z_{0})=\{x\in R^2:|x-z_{0}|<r\}$. Equations for $f(\xi,t)$, so called the Polubarinova-Galin equations,  are derived under this reformulation and they are expressed  in the case of injection or suction  respectively as follows:
 \begin{equation}
\label{PG1}
Re\left[f_{t}(\xi,t)\overline{f^{'}(\xi,t)\xi}\right]=1, \quad\xi\in\partial D
\end{equation}
and 
 \begin{equation}
\label{PGS}
Re\left[f_{t}(\xi,t)\overline{f^{'}(\xi,t)\xi}\right]=-1, \quad\xi\in\partial D.
\end{equation}
A solution to equation $(\ref{PG1})$ or $(\ref{PGS})$ is said to be a strong solution for $t\in [0,b)$ if $f(\xi,t)$ is univalent and analytic in a neighborhood of $\overline{D}$, $f(0,t)=0$, $f^{'}(0,t)>0$ and $f(\xi,t)$ is continously differentiable in $t\in [0,b)$.\par
 Denote 
 \begin{align}
 H(E)&=\left\{f\mid \mbox {$f(\xi)$ is analytic in $E$}\right\},\notag\\
 O(E)&=\left\{f\in H(E)\mid \mbox {$f(\xi)$ is univalent, $f(0)=0, f^{'}(0)>0$}\right\}.\notag
 \end{align}
The short-time well-posedness of (\ref{PG1}) has been thoroughly explored. In Reissig and von Wolfersdorf~\cite{reissig}, the authors prove the existence and uniqueness of a short-time strong solution in $O(\overline{D})$ if the initial function is in $O(\overline{D})$. In Gustafsson~\cite{gustaf1}, the author proves that a strong solution to (\ref{PG1}) is a family of  polynomials of degree $k_{0}$ if its initial function in $O(\overline{D})$ is also a polynomial of degree $k_{0}$. These results all can be applied to (\ref{PGS}) as well even though the authors don't comment on that.\par 
In this paper, we first prove a perturbation theorem for the strong polynomial solutions to the Polubarinova-Galin equation~(\ref{PG1}) or (\ref{PGS}). Many properties for strong polynomial solutions are thoroughly known. This theorem enables us to explore the properties of evolution of perturbed polynomials which are nonpolynomial. We obtain two applications of this theorem in the suction and injection case. \par
We first state this perturbation theorem. We define the following norms to describe the evolution of solutions:
\[
\left| \sum_{i=0}^{\infty}a_{i}\xi^{i}\right|_{M} =\sum_{i=0}^{\infty}\left|a_{i}\right|, \quad\left|\sum_{i=0}^{\infty}a_{i}\xi^{i}\right|_{M(r)}=\sum_{i=0}^{\infty}\left| a_{i}r^{i}\right|.
\]
Also, we define the following norm to describe the small perturbation:
 \[\|v\|_{\rho,n}=\sum_{j=1}^{\infty}\left| v_{j}\right|\rho^{j}j^{\frac{1}{2}+n}, \quad v=\sum_{j=1}^{\infty}v_{j}\xi^{j}.\]   
The perturbation theorem, Theorem~\ref{twins}, describes the evolution of small perturbation of polynomials and is stated as follows:\\
\begin{thm}\label{twins}
Given a strong degree $k_{0}$ polynomial solution $f_{k_{0}}(\xi,t)$ to (\ref{PG1}) (or (\ref{PGS})), and that $f_{k_{0}}(\xi,t)\in O(\overline{D_{r}})$ at $t\in [0,T_{0}]$ for some $T_{0}>0$ and $r>1$. Then for $\epsilon>0, k\in N$ and $1<r^{'}<r$, there exist $\delta(f_{k_{0}}, T_{0}, \epsilon, k, r')>0$ and $\rho(f_{k_{0}}, T_{0}, \epsilon, k, r')>1$ such that if 
$\left\|f(\cdot,0)-f_{k_{0}}(\cdot,0)\right\|_{\rho,k}<\delta$ where $f(0,0)=0$ and $f^{'}(0,0)>0$,  then the strong solution to (\ref{PG1}) (or (\ref{PGS})) $f(\xi,t)$ satisfies
\[f(\xi,t)\in O(\overline{D_{r'}})\cap C^{1}([0,T_{0}],H(D_{r})),\]
and for $0\leq n\leq k$, $0\leq t\leq T_{0},$
\[\left| f_{k_{0}}^{(n)}(\cdot,t)-f^{(n)}(\cdot,t)\right|_{M(r)}<\epsilon.\]
\end{thm}
The applications of this theorem and related past results are stated briefly in 1.1 and 1.2 as the following:
\par
\textbf{1.1} Here we assume the driving mechanism is suction.  It has been known that strong solutions to (\ref{PGS}) must blow up before the fluid is sucked out except for the degree $1$ polynomial solutions. However, by taking $k_{0}=1$ in Theorem~\ref{twins}, we prove that if the initial domain is close to a disk, most of fluid is sucked before the strong solution to (\ref{PGS}) blows up.\par
\textbf{1.2} Now we assume the driving mechanism is injection. In Sakai~\cite{sakai} and Gustafsson and Sakai \cite{sakai2}, the authors consider solutions of weak formulation and investigate the radius and curvature of two-dimensional moving domains respectively.  For an arbitrary initial shape the moving domain its asymptote is expanding disks. Recently, progress regarding this asymptotic behavior has been made by investigating it in terms of conserved quantities, so called Richardson complex moments; see Richardson~\cite{richardson}. In Vondenhoff~\cite{vondenhoff}, by restricting multi-dimensional initial domains to be close to balls, the author gives a rescaling behavior of the moving boundaries in terms of conserved moments.  In this paper, we aim to generalize the former result in two-dimensions by assuming a larger set of initial domains.\par
It has been known that there is a general class of polynomials which can give rise to global strong polynomial solutions to (\ref{PG1}) and the corresponding initial domains can be quite different from disks; for examples, starlike polynomials (eg. $\xi+\frac{2}{5}\xi^2$ and $\frac{\xi}{1.1}-\frac{15}{14}(\frac{\xi}{1.1})^2+\frac{4}{7}(\frac{\xi}{1.1})^3-\frac{1}{7}(\frac{\xi}{1.1})^4$); see Gustafsson, Prokhorov and Vasil'ev~\cite{gustaf2}. An arbitrary global strong degree $k_{0}$ polynomial solution to (\ref{PG1}), called $f_{k_{0}}(\xi,t)$, can have its  rescaling behaviors precisely described in terms of moments; see  Lin~\cite{Lin}. In this paper, as an application of Theorem~\ref{twins}, we show that small perturbation of $f_{k_{0}}(\xi,0)$, called $f(\xi,0)$, can give rise to a global strong solution $f(\xi,t)$ and a rescaling behavior of the corresponding moving domains, similar to that stated in Vondenhoff~\cite{vondenhoff},  is given in terms of moments as well. We can deduce  the case that initial domains are small perturbation of disks from this result by letting $k_{0}=1$. Therefore, this result generalizes the result in Vondenhoff~\cite{vondenhoff}. Lin~\cite{Lin}, Vondenhoff~\cite{vondenhoff} and this paper consider different sets of initial data and the rescaling behavior in Lin~\cite{Lin} is different from that in Vondenhoff~\cite{vondenhoff} and this paper. However, geometrically, these rescaling behaviors in the three work all imply that by rescaling the corresponding moving domain $\Omega(t), t\geq 0$ to be a domain $\Omega'(t)$ with area $\pi$, the radius and curvature of $\partial\Omega'(t)$ decay to $1$ algebraically and the decay is faster if lower moments vanish.\par
The sketch of proof of  this result is as the following: We first apply Theorem~\ref{twins} and prove the existence of a locally-in-time  strong solution $\{f(\xi,t)\}_{0\leq t\leq T_{0}}$ where $f(\xi,T_{0})$ is strongly starlike and $f(D,T_{0})$ is a small perturbation of a disk, even though $f(\xi,0)$ can be nonstarlike and $f(D,0)$ is far from a disk. Since starlikeness is a sufficient condition for an initial function to give rise to a global strong solution as shown in Gustafsson, Prokhorov and Vasil'ev~\cite{gustaf2} and since large-time rescaling behavior for evolution of  perturbed disks is shown in Vondenhoff~\cite{vondenhoff} in terms of moments, the solution $f(\xi,t)$ must be global and a rescaling behavior is given in terms of moments as well. 
\par
 The structure of this paper is as follows.  In \textbf{Section~\ref{sec4}}, we prove Theorem~\ref{twins}. In \textbf{Section~\ref{sec9}}, the application of Theorem~\ref{twins} in the suction case is given. In \textbf{Section~\ref{sec5}}, the application of Theorem~\ref{twins} in the injection case is given. As a byproduct of a theorem in this paper, a short proof of existence and uniqueness of strong solutions to (\ref{PG1}) is given in \textbf{Section~\ref{sec6}}.

\section{Proofs of Theorem~\ref{twins}}
\label{sec4}
The proof of the perturbation theorem in the suction case  is almost the same as the proof in the injection case. Therefore, we will just provide the proof of the theorem in the case of injection~(\ref{PG1}). \par
As in Gustafsson~\cite{gustaf1}, a reformulation of  the Polubarinova-Galin equation (\ref{PG1})  is expressed:
\begin{equation}
\label{PG2}
f_{t}=\xi f^{'}P\left[\frac{1}{\mid f^{'}\mid^2}\right],\quad\xi\in D
\end{equation}
where $P$ denotes the Poisson kernel  which defines the analytic function in the unit disk
\begin{equation}
\label{poisson100}
P\left[ g \right](\xi) =  \frac{1}{2\pi i}\int_{\partial D}g(z)\frac{z+\xi}{z-\xi}\frac{dz}{z}, \quad \xi \in D,
\end{equation}
from boundary data $g$ on $\partial D$.  In the mathematical treatment of (\ref{PG2}) it makes no difference if $f(\xi,t)$ is univalent in $\overline{D}$ or merely locally univalent in $\overline{D}$; see Gustafsson~\cite{gustaf1}. To make a distinction, we denote 
\[
         \omega(E)=\left\{f\in H(E)\mid \mbox{$f$ is locally univalent in $E$, $f(0)=0$ and $ f^{'}(0)>0$}\right\}\]
and  define a solution to be a strong* solution to (\ref{PG2}) as follows:\\

\begin{defn}
A solution $f(\xi,t)\in\omega(\overline{D})$ is a strong* solution to (\ref{PG2}) for $0\leq t< b$ if $f(\xi,t)$ is continuously differentiable with respect to $t\in [0,b)$ and satisfies (\ref{PG2}).
\end{defn} 
An univalent strong* solution $f(\xi,t)$ to $(\ref{PG2})$ must be a strong solution to the Polubarinova-Galin equation (\ref{PG1}).\par
In subsection \ref{aaa}, we aim to prove a perturbation theorem for strong* polynomial solutions to (\ref{PG2}), Theorem~\ref{Main Lemma}. In subsection \ref{univalentthm}, we show that Theorem~\ref{twins} follows directly from Theorem~\ref{Main Lemma}.
\subsection{A perturbation theorem for strong* polynomial solutions}
\label{aaa}
We start with lemmas before proving the perturbation theorem for strong* polynomial solutions to (\ref{PG2}).
\begin{lemma}
\label{Lemma2.2}
For $1<p<\infty$, there exists $C_{p}>0$ such that
\[\left\|P\left[g\right]\right\|_{L^p([0,2\pi])}\leq C_{p}\|g\|_{L^p([0,2\pi])}\]
 for $g$ which  is  holomorphic in a neighborhood of $\partial D$ and is also a real function on $\partial D$.
 \end{lemma}
\begin{proof}
There exists $u$ which is harmonic in $D$, continuous in $\overline{D}$, and $u=g$ on $\partial D$. Therefore, by Theorem 17.26 in Rudin~\cite{rudin}, it is shown that for $1<p<\infty$, there exists $C_{p}>0$ such that 
\[\left\|P\left[u\right]\right\|_{L^p([0,2\pi])}\leq C_{p}\|u\|_{L^p([0,2\pi])},\]
which means
\[\left\|P\left[g\right]\right\|_{L^p([0,2\pi])}\leq C_{p}\|g\|_{L^p([0,2\pi])}.\]
\end{proof}
In the proof of the perturbation theorem for strong* polynomial solutions, we use iterative methods. In each iteration, we need to calculate the difference of two polynomial univalent functions $h_{1}$ and $ h_{2}$ which satisfy the assumption of Lemma~\ref{estimate11}. Inequality (\ref{DDD}) enables us to estimate $\|h_{1}^{'}-h_{2}^{'}\|_{L^{2}([0,2\pi])}$ locally in time when $h_{1}$ and $h_{2}$ are both polynomial as shown in the proof of Theorem~\ref{Main Lemma}.
\begin{lemma}
\label{estimate11}
Let $g(\xi,t)\in \omega(\overline{D_{r}})\cap C^{1}([0,t_{1}],H(\overline{D_{r}}))$ be a strong* solution to (\ref{PG2}) and $0<l<1$. There exists $C(g, t_{1}, r, l)>0$ such that, if  $h_{1}(z,t), h_{2}(z,t)$ $\in\omega(\overline{D_{r}})\cap C^{1}([0,t_{h}],H(\overline{D_{r}}))$ are two strong* solutions to (\ref{PG2})  where $0<t_{h}\leq t_{1}$ and 
\begin{equation}
\label{gg}
\max_{([0,t_{h}])}\left| h_{i}^{'}(\cdot,t)-g^{'}(\cdot,t)\right|_{M(r)}\leq l\min_{(\overline{D_{r}},[0,t_{1}])}\left| g^{'}\right|,\quad 1\leq i\leq 2,
\end{equation}
then we have
\begin{equation}
\label{iterative}
\left\|\frac{\partial}{\partial t}\left[h_{1}-h_{2}\right]\right\|_{L^{2}([0,2\pi])}\leq C\left\|h_{1}^{'}-h_{2}^{'}\right\|_{L^{2}([0,2\pi])},\quad 0\leq t\leq t_{h}.
\end{equation}
Furthermore, if $h_{1}, h_{2}$ are both polynomials of degree $\leq n$, then for $0\leq t\leq t_{h}$
\begin{equation}
\label{DDD}
\left\|h_{1}^{'}(\cdot,t)-h_{2}^{'}(\cdot,t)\right\|_{L^{2}([0,2\pi])}^2\leq e^{2C(n)t}\left\|h_{1}^{'}(\cdot,0)-h_{2}^{'}(\cdot,0)\right\|_{L^{2}([0,2\pi])}^2.
\end{equation}
\end{lemma}
\begin{proof}
(1)\begin{equation}
\label{B1}
\frac{\partial}{\partial t}[h_{1}-h_{2}]=\xi\left\{\left[h_{1}^{'}-h_{2}^{'}\right]P\left[\frac{1}{\left| h_{2}^{'}\right|^2}\right]+h_{1}^{'}P\left[\frac{1}{\left| h_{1}^{'}\right|^2}-\frac{1}{\left| h_{2}^{'}\right|^2}\right]\right\}.
\end{equation}
Here, by Lemma~\ref{Lemma2.2},
\begin{equation}
\label{B2}
\left\|P\left[\frac{1}{\left| h_{1}^{'}\right|^2}-\frac{1}{\left| h_{2}^{'}\right|^2}\right]\right\|_{L^2([0,2\pi])}\leq C_{2}\left\|\frac{1}{|h_{1}^{'}|^2}-\frac{1}{|h_{2}^{'}|^2}\right\|_{L^2([0,2\pi])}.
\end{equation}
By taking the $L_{2}$ norm of the right-hand side  and the left-hand side of (\ref{B1}) and then using (\ref{B2}) and H$\ddot{o}$lder's inequality, we obtain
\begin{align}
\label{calp}
&\left\|\frac{\partial}{\partial t}[h_{1}-h_{2}]\right\|_{L^{2}([0,2\pi])}\notag\\
\leq&\left\|h_{1}^{'}-h_{2}^{'}\right\|_{L^{2}([0,2\pi])}\max_{\partial D}\left|P\left[\frac{1}{\left| h_{2}^{'}\right|^2}\right]\right|+C_{2}\left\|h_{1}^{'}\right\|_{L^{\infty}([0,2\pi])}\left\|\frac{1}{| h_{1}^{'}|^2}-\frac{1}{| h_{2}^{'}|^2}\right\|_{L^{2}([0,2\pi])}\notag\\
\leq &\left\{\max_{\partial D}\left|P\left[\frac{1}{\left| h_{2}^{'}\right|^2}\right]\right|+C_{2}\max_{\partial D}\left| h_{1}^{'}\right|\max_{\partial D}\frac{\left|h_{1}^{'}\right|+\left|h_{2}^{'}\right|}{\left|h_{1}^{'}\right|^2 \left|h_{2}^{'}\right|^2}\right\}\left\|h_{1}^{'}-h_{2}^{'}\right\|_{L^{2}([0,2\pi])}
\end{align}
We want to bound 
\[\max_{\partial D}\left| h_{1}^{'}\right|\max_{\partial D}\frac{\left|h_{1}^{'}\right|+\left|h_{2}^{'}\right|}{\left|h_{1}^{'}\right|^2 \left|h_{2}^{'}\right|^2}\mbox{\quad and\quad}\max_{\partial D}\left|P\left[\frac{1}{\left| h_{2}^{'}\right|^2}\right]\right| \]
respectively in (i) and (ii) in terms of $g$ and hereby determine the constant $C$.\\ 
(i)By assumption (\ref{gg}),  for $(z,t)\in (\partial D,[0,t_{h}])$,
\begin{equation}
\label{estimate1}
\left| h_{i}^{'}(z,t)\right|\geq\left|g^{'}(z,t)\right|-\left|h_{i}^{'}(z,t)-g^{'}(z,t)\right|\geq(1-l)\left| g^{'}(z,t)\right|, \quad1\leq i\leq 2
\end{equation}
and 
\begin{equation}
\label{estimate10}
\left| h_{i}^{'}(z,t)\right|\leq\left|g^{'}(z,t)\right|+\left|h_{i}^{'}(z,t)-g^{'}(z,t)\right|\leq(1+l)\left| g^{'}(z,t)\right|, \quad1\leq i\leq 2.
\end{equation}
Therefore, by (\ref{estimate1}) and (\ref{estimate10}), for $0\leq t\leq t_{h}$
\begin{equation}
\label{cals}
\max_{\partial D}\left| h_{1}^{'}\right|\max_{\partial D}\frac{|h_{1}^{'}|+|h_{2}^{'}|}{|h_{1}^{'}|^2 |h_{2}^{'}|^2}\leq 2\frac{(1+l)}{(1-l)^3}\max_{(\partial D,[0,t_{1}])}\left| g^{'}\right|\max_{(\partial D,[0,t_{1}])}\frac{1}{\left| g^{'}\right|^3}.\end{equation}
(ii)We start with finding the upper bound of $P\left[\frac{1}{\left| h_{2}^{'}\right|^2}-\frac{1}{\left| g^{'}\right|^2}\right]$ in terms of $g$ and hereby obtain the upper bound for $P\left[\frac{1}{|h_{2}^{'}|^2}\right]$ in terms of $g$. \par
In Gustafsson~\cite{gustaf1},  it is shown that for given $h\in \omega(\overline{D_{r}})$, 
\begin{equation}
\label{enlarge1}
P\left[\frac{1}{\left| h^{'}\right|^2}\right]=\frac{1}{2\pi i}\int_{\partial D_{r}}\frac{1}{h^{'}(z,t)\overline{h^{'}}(1/z,t)}\frac{z+\xi}{z-\xi}\frac{dz}{z},\quad\xi\in D.
\end{equation}
By (\ref{enlarge1}), we have for $\xi\in D$
\begin{equation}
\label{cal1}
P\left[\frac{1}{\left| h_{2}^{'}\right|^2}-\frac{1}{\left| g^{'}\right|^2}\right]=\frac{1}{2\pi i}\int_{\partial D_{r}}\left(\frac{1}{h_{2}^{'}(z,t)\overline{h_{2}^{'}}(1/z,t)}-\frac{1}{g^{'}(z,t)\overline{g^{'}}(1/z,t)}\right)\frac{z+\xi}{z-\xi}\frac{dz}{z}.
\end{equation}
Therefore, 
\begin{align}
\label{call}
&\max_{\partial D}\left| P\left[\frac{1}{\left| h_{2}^{'}\right|^2}-\frac{1}{\left| g^{'}\right|^2}\right]\right|\notag\\
\leq&\max_{\partial D_{r}}\left|\frac{1}{h_{2}^{'}(z,t)\overline{h_{2}^{'}}(1/z,t)}-\frac{1}{g^{'}(z,t)\overline{g^{'}}(1/z,t)}\right|\frac{r+1}{r-1}\notag\\
=&\max_{\partial D_{r}}\left|\frac{h_{2}^{'}(z,t)-g^{'}(z,t)}{g^{'}(z,t)\overline{g^{'}}(1/z,t)h_{2}^{'}(z,t)}+\frac{\overline{h_{2}^{'}}(1/z,t)-\overline{g^{'}}(1/z,t)}{h_{2}^{'}(z,t)\overline{h_{2}^{'}}(1/z,t)\overline{g^{'}}(1/z,t)}\right|\frac{r+1}{r-1}.\notag\\
\end{align}
By assumption (\ref{gg}), for $(z,t)\in (\partial D_{r},[0,t_{h}])$,
\begin{equation}
\label{estimate20}
\left|\overline{h_{2}^{'}}(1/z,t)- \overline{g^{'}}(1/z,t)\right|\leq l\left|\overline{g^{'}}(1/z,t)\right|,
\end{equation}
\begin{equation}
\label{estimate2}
\left|\overline{h_{2}^{'}}(1/z,t)\right|\geq(1-l)\left| \overline{g^{'}}(1/z,t)\right|.
\end{equation}
By assumption (\ref{gg}), for $(z,t)\in (\partial D_{r},[0,t_{h}])$,
\begin{equation}
\label{estimate60}
\left|h_{2}^{'}(z,t)- g^{'}(z,t)\right|\leq l\left|g^{'}(z,t)\right|,
\end{equation}
\begin{equation}
\label{estimate6}
\left|h_{2}^{'}(z,t)\right|\geq(1-l)\left|g^{'}(z,t)\right|.
\end{equation}
By (\ref{estimate20})-(\ref{estimate6}), 
\[\max_{(\partial D_{r},[0,t_{h}])}\left|\frac{h_{2}^{'}(z,t)-g^{'}(z,t)}{g^{'}(z,t)\overline{g^{'}}(1/z,t)h_{2}^{'}(z,t)}+\frac{\overline{h_{2}^{'}}(1/z,t)-\overline{g^{'}}(1/z,t)}{h_{2}^{'}(z,t)\overline{h_{2}^{'}}(1/z,t)\overline{g^{'}}(1/z,t)}\right|\frac{r+1}{r-1}\]
\[\leq 2l\left[\max_{(\partial D_{r},[0,t_{1}])}\left| \frac{1}{\left|g^{'}(z,t)\right|\left|\overline{g^{'}}(1/z,t)\right|(1-l)^2}\right|\right]\frac{r+1}{r-1}.\]
Therefore, by the above inequality and (\ref{call}), for $0\leq t\leq t_{h}$
\begin{align}
&\max_{\partial D}\left|\xi P\left[\frac{1}{\left| h_{2}^{'}\right|^2}-\frac{1}{\left| g^{'}\right|^2}\right]\right|\notag\\
\leq& 2l\left[\max_{(\partial D_{r},[0,t_{1}])}\left| \frac{1}{\left|g^{'}(z,t)\right|\left|\overline{g^{'}}(1/z,t)\right|(1-l)^2}\right|\right]\frac{r+1}{r-1}.\notag\\
\end{align}
Hence, for $0\leq t\leq t_{h}$, we have
\begin{align}
&\max_{\partial D}\left|\xi P\left[\frac{1}{\left| h_{2}^{'}\right|^2}\right]\right|\notag\\
\leq&\max_{\partial D}\left|\xi P\left[\frac{1}{\left| h_{2}^{'}\right|^2}-\frac{1}{\left| g^{'}\right|^2}\right]\right|+\max_{\partial D}\left| \xi P\left[\frac{1}{\left| g^{'}\right|^2}\right]\right|\notag\\
\leq&2l\left[\max_{(\partial D_{r},[0,t_{1}])}\left| \frac{1}{\left|g^{'}(z,t)\right|\left|\overline{g^{'}}(1/z,t)\right|(1-l)^2}\right|\right]\frac{r+1}{r-1}+\max_{(\partial D,[0,t_{1}])}\left| \xi P\left[\frac{1}{\left| g^{'}\right|^2}\right]\right|.
\end{align}
From (i) and (ii), we prove (\ref{iterative}) by choosing $C$ to be
\begin{align}
C=&\max_{(\partial D,[0,t_{1}])}\left| \xi P\left[\frac{1}{\left| g^{'}\right|^2}\right]\right|+2l\left[\max_{(\partial D_{r},[0,t_{1}])}\left| \frac{1}{|g^{'}(z,t)||\overline{g^{'}}(1/z,t)|(1-l)^2}\right|\right]\frac{r+1}{r-1}\notag\\
&+C_{2}\left\{2\frac{(1+l)}{(1-l)^3}\max_{(\partial D,[0,t_{1}])}\left| g^{'}\right|\max_{(\partial D,[0,t_{1}])}\frac{1}{\left| g^{'}\right|^3}\right\}.
\end{align}
\\
(2)Now we assume that $h_{1}, h_{2}$ are both polynomials of degree $\leq n$. Denote $h_{1}=\sum_{i=1}^{n}\alpha_{i}(t)\xi^{i}$ and $h_{2}=\sum_{i=1}^{n}\beta_{i}(t)\xi^{i}$. Also denote $D(t)$ by
\[D(t)
=\left\|h_{1}^{'}-h_{2}^{'}\right\|_{L^{2}([0,2\pi])}^2=2\pi\left\{\left(\sum_{i=1}^{n}[\left| \alpha_{i}(t)-\beta_{i}(t)\right|]^2i^2\right)\right\}.\]
Then
\begin{align}
          D^{'}(t)&=2\pi\cdot 2\left\{\left(\sum_{i=1}^{n}Re\left[(\alpha_{i}-\beta_{i})\overline{(\alpha_{i}-\beta_{i})_{t}}\right]i^2\right)\right\}\notag\\
&\leq 2\pi\cdot 2(n)\left\{\left(\sum_{i=1}^{n}\left|(\alpha_{i}-\beta_{i})\right|\left|(\alpha_{i}-\beta_{i})_{t}\right| i\right)\right\}\notag\\
&\leq 2\pi\cdot 2(n)\left\{\left(\sum_{i=1}^{n}\left|(\alpha_{i}-\beta_{i})\right|^2i^2\right)\right\}^{\frac{1}{2}}\left\{\left(\sum_{i=1}^{n}\left|(\alpha_{i}-\beta_{i})_{t}\right|^2\right)\right\}^{\frac{1}{2}}\notag\\
&=2(n)\left\|\left[h_{1}^{'}-h_{2}^{'}\right]\right\|_{L^{2}([0,2\pi])}\left\|\frac{\partial}{\partial t}\left[h_{1}-h_{2}\right]\right\|_{L^{2}([0,2\pi])}.\notag
 \end{align}
By applying (\ref{iterative}) to the above inequality, we conclude that for $0\leq t\leq t_{h}$,
\begin{equation}
\label{333}
D^{'}(t)\leq 2C(n)\left\|\left[h_{1}^{'}-h_{2}^{'}\right]\right\|_{L^{2}([0,2\pi])}^2=2C(n)D(t),
\end{equation}
and therefore
\begin{equation}
\label{444}
D(t)\leq D(0)e^{2Ct(n)},
\end{equation}
which proves (\ref{DDD}).
\end{proof}
The following lemma helps us to control the blow-up time of strong* polynomial solutions to (\ref{PG2}).
\begin{lemma}
\label{Lemma2.4}Given a polynomial mapping $f(\xi,0)\in \omega(\overline{D_{r_{0}}})$ for some $r_{0}>1$, then there exists a unique strong* polynomial solution
 to $(\ref{PG2})$ $f(\xi,t)\in \omega(\overline{D_{r_{0}}})$ at least for a short time. Furthermore, if  the strong* polynomial solution ceases to exist at $t=b$,  then for any $r>1$, 
 \begin{equation}
 \label{blowuptime}
\lim\inf_{t\rightarrow b}\left(\min_{\overline{D_{r}}}\left| f^{'}(\xi,t)\right|\right)=0.
\end{equation}
\end{lemma}
\begin{proof}
(a)The first part follows from Gustafsson~\cite{gustaf1}.\\
(b)Assume that (\ref{blowuptime}) does not hold now. Then there exists $r>1$ such that 
\[\lim\inf_{t\rightarrow b}\left(\min_{\overline{D_{r}}}\left| f^{'}(\xi,t)\right|\right)>0.\]
This implies that there exist $C>0$ and $1<r'\leq r$ such that
\[\min_{\overline{D_{r'}}}\left|f^{'}(\xi,t)\right|>C, \quad t\in [0,b).\]
Since each coefficient of $f(\xi,t)$ is bounded for  $t\in [0,b)$, there exists $M>0$ such that
\[\sup_{t\in [0,b)}\max_{\overline{D_{r'}}}\left| f^{'}(\xi,t)\xi\right|\leq M.\]
For $\xi\in\overline{D}$
\begin{align}
&\sup_{t\in [0,b)}\left|f^{'}(\xi,t)\xi P\left[\frac{1}{\left| f^{'}\right|^2}\right]\right|\notag\\
\leq&\sup_{t\in [0,b)}\left|\frac{f^{'}(\xi,t)\xi}{2\pi i}\int_{\partial D_{r'}}\frac{1}{ f^{'}(z,t)\overline{f^{'}}(1/z,t)}\frac{z+\xi}{z-\xi}\frac{dz}{z}\right|\notag\\
\leq&\sup_{t\in [0,b)}\left(\max_{\overline{D}}\left| f^{'}(\xi,t)\xi\right|\cdot\max_{\partial D_{r'}}\left|\frac{1}{f^{'}(z,t)\overline{f^{'}}(1/z,t)}\right|\frac{r'+1}{r'-1}\right)\notag\\
\leq &\frac{M}{C^2}\frac{r'+1}{r'-1}\notag
\end{align}
Therefore, for $0\leq t_{2}<t_{1}<b$, $\xi\in D$
\[\left| f(\xi,t_{1})-f(\xi,t_{2})\right|=\left|\int_{t_{2}}^{t_{1}}f^{'}(\xi,t)\xi P\left[\frac{1}{\left| f^{'}\right|^2}\right]dt\right|\leq\left| t_{1}-t_{2}\right|\frac{M}{C^2}\frac{r'+1}{r'-1}.\]
Therefore $\lim_{t\rightarrow b}f(\xi,t)$ exists and we define it as $f(\xi,b).$ Note that $f(\xi,b)$
satisfies $\min_{\overline{D_{r'}}}\left| f^{'}(\xi,b)\right|\geq C$. Let $f(\xi,t+b)$ be the strong* solution to (\ref{PG2})
with the initial value $f(\xi,b)$ for $t\in [0,\epsilon)$. Then $f(\xi,t)$ is  continuous with respect to $t$ 
for $t\in [0,b+\epsilon)$ and 
\[f(\xi,t)-f(\xi,0)=\int_{0}^{t}f^{'}(\xi,s)\xi P\left[\frac{1}{\left| f^{'}(\cdot,s)\right|^2}\right]ds.\]
This implies that $f(\xi,t)\in\omega(\overline{D})$ is continuously differentiable with respect to $t$ for $t\in [0,b+\epsilon)$ and satisfies (\ref{PG2}). Hence it is impossible that $f(\xi,t)$ ceases to exist  at $t=b$ and therefore for any $r>1$, 
\[\lim\inf_{t\rightarrow b}\left(\min_{\overline{D_{r}}}\left| f^{'}(\xi,t)\right|\right)=0.\]
\end{proof}
\begin{thm}\label{Main Lemma}
Assume that $f_{k_{0}}(\xi,t)\in C^{1}([0,t_{1}],H(\overline{D_{r}}))\cap \omega(\overline{D_{r}})$ is a strong* degree $k_{0}$ polynomial solution to $(\ref{PG2})$ for some $t_{1}>0$ and $r>1$ and that $\rho>r$ and $l< 1$. If $f(\xi,0)$ satisfies the assumption\\
(A)\[\left\|f(\xi,0)-f_{k_{0}}(\xi,0)\right\|_{\rho,1}\leq\frac{l}{\sqrt{k_{0}}}\min_{(\overline{D_{r}},[0,t_{1}])}\left|f_{k_{0}}^{'}\right|\]
where $f^{'}(0,0)\in R$ and $f(0,0)=0$, then the following (a)-(b) are true:\\
(a)There exists $C(f_{k_{0}}, t_{1}, r, l)>0$ such that a strong* solution to (\ref{PG2}) $f(\xi,t)\in C^{1}([0,t_{0}],H(D_{r})\cap C(\overline{D_{r}}))\cap \omega(D_{r})$ where $t_{0}=\min\Big\{\frac{1}{Ck_{0}}(\ln\frac{\rho}{r}),t_{1}\Big\}$. Moreover,
\[\max_{([0,t_{0}])}\left| f^{'}-f^{'}_{k_{0}}\right|_{M(r)}\leq l\min_{(\overline{D_{r}},[0,t_{1}])}\left| f_{k_{0}}^{'}\right|.\]
(b)Furthermore, if there exist $\delta>0$ and $j$ nonnegative integer such that 
\[\left\|f(\cdot,0)-f_{k_{0}}(\cdot,0)\right\|_{\rho,j}\leq\delta,\]
then there exists $c(j, k_{0})>0$ such that
\[\max_{([0,t_{0}])}\left|f^{(j)}-f_{k_{0}}^{(j)}\right|_{M(r)}\leq c(j, k_{0})\delta.\]
\end{thm}
\begin{rem}The strong* solution $f(\xi,t)$ is obtained by using many polynomial strong* solutions to (\ref{PG2}) to approximate it. \end{rem}
\begin{proof}
(a)We take the constant $C(f_{k_{0}}, t_{1}, r, l)$ in (a) to be the same as the one defined in Lemma~\ref{estimate11}. We want to prove (a) in the following, by showing that there exists a strong* solution $f(\xi,t)\in \omega(D_{r})$ to (\ref{PG2}) for $0\leq t\leq t_{0}$, where $f(\xi,0)=f_{k_{0}}(\xi,0)+\sum_{i=1}^{\infty}b_{i}(0)\xi^{i}$ and 
\begin{equation}
\label{cal2}
\sum_{k=1}^{\infty}\left| b_{k}(0)\right| \rho^{k}k^{3/2}\leq\frac{l}{\sqrt{k_{0}}}\min_{(\overline{D_{r}},[0,t_{1}])}\left| f_{k_{0}}^{'}\right|.
\end{equation}
Denote the strong* polynomial solution to (\ref{PG2}) with the initial value $f_{k_{0}}(\xi,0)+\sum_{i=1}^{k}b_{i}(0)\xi^{i}$ by $g_{k}(\xi,t)$. The proof for (a) is split into step1 and step2. In step1, we prove that $g_{k}(\xi,t), k\geq 1$ exists for $t\in [0,t_{0}]$. In step2, we prove that $g_{k}(\xi,t)$ converges to the strong* solution $f(\xi,t)$ as $k$ goes to infinity and that $f(\xi,t)$ exists for $t\in [0,t_{0}]$.\\
Step1: \\
By (\ref{cal2}), there exist $\{d_{k}\}_{k\geq 0}$  nonnegative and  $\sum_{k=0}^{\infty}d_{k}=1$ such that
$\mid b_{i}(0)\mid\leq M_{i}\rho^{-i}$ for $i\geq 1$ where
\[M_{k+1}\leq\frac{l}{\sqrt{k_{0}}}\frac{d_{k}}{(k+1)^{3/2}}\min_{(\overline{D_{r}},[0,t_{1}])}\left| f_{k_{0}}^{'}\right|, \quad k\geq 0.\]
\textbf{Claim}:\\
 Prove that for $k\geq 0$, $g_{k}(\xi,t)\in C^{1}([0,t_{0}], H(\overline{D_{r}}))\cap \omega(\overline{D_{r}})$ and
\[\max_{([0,t_{0}])}\left| g_{k}^{'}-g_{k+1}^{'}\right|_{M(r)}\leq ld_{k}\min_{(\overline{D_{r}},[0,t_{1}])}\left| g_{0}^{'}\right|.\]
\begin{proof}(proof of claim)
We prove it by induction as follows.\\
(i)Assume for $0\leq k\leq n-1$,
\[\max_{([0,t_{0}])}\left| g_{k}^{'}-g_{k+1}^{'}\right|_{M(r)}\leq ld_{k}\min_{(\overline{D_{r}},[0,t_{1}])}\left| g_{0}^{'}\right|.\]
(ii){\bf Subclaim}: \\
For $t\in [0,t_{0}]$
\begin{equation}
\label{induction2}
\left| g_{n}^{'}-g_{n+1}^{'}\right|_{M(r)}\leq ld_{n}\min_{(\overline{D_{r}},[0,t_{1}])}\left| g_{0}^{'}\right|.
\end{equation}
\begin{proof}{(of subclaim)}
Denote $s_{n}=\sup\{T\leq t_{0}|\mbox{$g_{n+1}(\xi,t)$ satisfies (\ref{induction2}) for $t\in [0,T]$} \}$. Then $|g_{n+1}'|\geq (1-l)|g_{0}'|$ for $t\in [0,s_{n})$. Therefore,  by Lemma~\ref{Lemma2.4}, the value $s_{n}=\max\{T\leq t_{0}|\mbox{$g_{n+1}(\xi,t)$ satisfies (\ref{induction2}) for $t\in [0,T]$} \}$.\par
For $0<t\leq s_{n}$, 
\begin{equation}
\label{induction3}
\max_{([0,t])}\left| g_{n+1}^{'}-g_{0}^{'}\right|_{M(r)}\leq\sum_{k=0}^{n}ld_{k}\min_{(\overline{D_{r}},[0,t_{1}])}\left| g_{0}^{'}\right|\leq l\min_{(\overline{D_{r}},[0,t_{1}])}\left| g_{0}^{'}\right|.
\end{equation}
Also by the assumption in (i), we have
\begin{equation}
\label{induction1}
\max_{([0,t_{0}])}\left| g_{n}^{'}-g_{0}^{'}\right|_{M(r)}\leq\sum_{k=0}^{n-1}ld_{k}\min_{(\overline{D_{r}},[0,t_{1}])}\left| g_{0}^{'}\right|\leq l\min_{(\overline{D_{r}},[0,t_{1}])}\left| g_{0}^{'}\right|.
\end{equation}
From (\ref{induction3}) and (\ref{induction1}),  $g_{0}$, $g_{n}$ and $g_{n+1}$ satisfy the assumption for $g$, $h_{1}$ and $h_{2}$ in Lemma~\ref{estimate11} respectively. Denote  $D(t)=\left\|g_{n+1}^{'}-g_{n}^{'}\right\|_{L^{2}([0,2\pi])}^2$. From Lemma~\ref{estimate11}, we can obtain that for $0\leq t\leq s_{n}$
\begin{equation}
\label{inequality}
D(t)\leq e^{2C(n+1)k_{0}t}D(0).
\end{equation}
{\bf We need to show $s_{n}=t_{0}$}.\par
Note that if $s_{n}<t_{0}$, then the following $(R_{1})$ must hold:\\
($R_{1}$) At time $t=s_{n}$,
\[\left| g_{n}^{'}-g_{n+1}^{'}\right|_{M(r)}=d_{n}\min_{(\overline{D_{r}},[0,t_{1}])}\left| g_{0}^{'}\right|l.\]
Assume that $s_{n}< t_{0}$ now. Then for $0\leq t\leq s_{n}$,
\begin{align}
&\left| g_{n}^{'}-g_{n+1}^{'}\right|_{M(r)}\notag\\
           \leq &\sqrt{D(t)(n+1)k_{0}}r^{(n)}\notag\\
           \leq &\sqrt{(n+1)k_{0}D(0)e^{2Ctk_{0}(n+1)}r^{2(n)}}\notag\\
           \leq &\sqrt{(n+1)k_{0}D(0)e^{2Cs_{n}k_{0}(n+1)}r^{2(n)}}\notag\\
           < &\sqrt{(n+1)k_{0}D(0)e^{2Ct_{0}k_{0}(n+1)}r^{2(n)}}.\notag
                \end{align} 
Since  
\[D(0)(n+1)k_{0}\leq(\rho)^{-2(n+1)}(d_{n})^2\min_{(\overline{D_{r}},[0,t_{1}])}\left| g_{0}^{'}\right|^2l^2,\]
we have
\begin{align}
&\max_{([0,s_{n}])}\left| g_{n}^{'}-g_{n+1}^{'}\right|_{M(r)}\notag\\
\leq &\sqrt{(n+1)k_{0}D(0)e^{2Ct_{0}k_{0}(n+1)}r^{2(n)}}\notag\\
< & d_{n}\min_{(\overline{D_{r}},[0,t_{1}])}\left| g_{0}^{'}\right|l\notag
\end{align}
which contradicts the remark ($R_{1}$). Therefore, $s_{n}=t_{0}$.\end{proof}
\end{proof}
Step2:\\
By Step 1, for $k\geq 1$
\[\max_{([0,t_{0}])}\left| g_{k}^{'}-g_{0}^{'}\right|_{M(r)}\leq l\sum_{n=0}^{\infty}d_{n}\min_{(\overline{D_{r}},[0,t_{1}])}\left| g_{0}^{'}\right|\leq l\min_{(\overline{D_{r}},[0,t_{1}])}\left| g_{0}^{'}\right|.\] 
There exists $f(\xi,t)\in C([0,t_{0}], \omega(D_{r})\cap C(\overline{D_{r}}))$ such that $| g^{'}_{k}-f^{'}|_{M(r)}$ goes to zero as $k$ goes to $\infty$. Furthermore, 
\[\max_{([0,t_{0}])}\left| f^{'}-g_{0}^{'}\right|_{M(r)}\leq l\min_{(\overline{D_{r}},[0,t_{1}])}\left|g_{0}^{'}\right|.\]
Still, we have to show that $f(\xi,t)$ satisfies $(\ref{PG2})$. Fix $1<r^{'}<r$. For $\xi\in D_{r'}$ and $0\leq t\leq t_{0}$,
\begin{equation}
\label{ppp}
\frac{\partial}{\partial t}g_{k}(\xi,t)=\frac{g_{k}^{'}(\xi,t)\xi}{2\pi i}\int_{\partial D_{r'}}\frac{1}{g_{k}^{'}(z,t)\overline{g_{k}^{'}}(1/z,t)}\frac{z+\xi}{z-\xi}\frac{dz}{z}.
\end{equation}
By integrating (\ref{ppp}) with respect to $t$, we have that for $\xi\in D_{r'}$ and $0\leq t\leq t_{0}$,
\[g_{k}(\xi,t)-g_{k}(\xi,0)=\int_{0}^{t}\frac{g_{k}^{'}(\xi,s)\xi}{2\pi i}\int_{\partial D_{r'}}\frac{1}{g_{k}^{'}(z,s)\overline{g_{k}^{'}}(1/z,s)}\frac{z+\xi}{z-\xi}\frac{dz}{z}ds.\]
Let $k\rightarrow\infty$. For $\xi$ in any compact subset of $D_{r'}$,
\begin{equation}
\label{continous} 
f(\xi,t)-f(\xi,0)=\int_{0}^{t}\frac{f^{'}(\xi,s)\xi}{2\pi i}\int_{\partial D_{r'}}\frac{1}{f^{'}(z,s)\overline{f^{'}}(1/z,s)}\frac{z+\xi}{z-\xi}\frac{dz}{z}ds
\end{equation}
for some $f(\xi,t)\in C([0,t_{0}], \omega(D_{r})\cap C(\overline{D_{r}}))$. Furthermore, the identity  $(\ref{continous})$ shows that $f(\xi,t)\in C^{1}([0,t_{0}], H(D_{r})\cap C(\overline{D_{r}}))$.\\
(b)Now assume (b). Then 
\[\left| b_{i}(0)\right|\leq M_{i}\rho^{-i},\quad i\geq 1\]
where
\[M_{k+1}\leq\frac{1}{(k+1)^{\frac{1}{2}+j}}d_{k}\delta,\quad k\geq 0.\]
First we look at the case $j=2$. Under (b),
\begin{align}
&\max_{([0,t_{0}])}\left| g_{n}^{''}-g_{n+1}^{''}\right|_{M(r)}\notag\\
\leq&\sqrt{(n+2)^{3}(k_{0}+1)^{3}\frac{1}{3}D(0)e^{2Ct_{0}k_{0}(n+1)}}r^{n-1}\notag\\
=&\left(\frac{n+2}{n+1}\right)^{\frac{3}{2}}\frac{1}{\sqrt{3}}(k_{0}+1)^{\frac{3}{2}}\sqrt{D(0)(n+1)^3e^{2ct_{0}k_{0}(n+1)}}r^{n-1}\notag\\
\leq &\left(\frac{n+2}{n+1}\right)^{\frac{3}{2}}\frac{1}{\sqrt{3}}(k_{0}+1)^{\frac{3}{2}}d_{n}\delta, \quad n\geq 0.\notag
                \end{align} 
Therefore, we have for $n\geq 1$
\[\max_{([0,t_{0}])}\left| g_{0}^{''}-g_{n}^{''}\right|_{M(r)}\leq\frac{1}{\sqrt{3}}2^{\frac{3}{2}}(k_{0}+1)^{\frac{3}{2}}\delta.\]
Similarly, for $j\geq 2$,  under the assumption of (b), there exists $c(j,k_{0})>0$ such that
\begin{align}
&\max_{([0,t_{0}])}\left| g_{n}^{(j)}-g_{n+1}^{(j)}\right|_{M(r)}\notag\\
\leq&c(j,k_{0})\sqrt{(n+1)^{2j-1}D(0)e^{2Ct_{0}k_{0}(n+1)}}\notag\\
\leq&c(j,k_{0})d_{n}\delta.\notag
     \end{align} 
Therefore, we have 
\[             
\max_{([0,t_{0}])}\left| g_{0}^{(j)}-g_{n}^{(j)}\right|_{M(r)}\leq c(j,k_{0})\delta.\]
Let $n\rightarrow\infty,$
\[\max_{([0,t_{0}])}\left| g_{0}^{(j)}-f^{(j)}\right|_{M(r)}\leq  c(j,k_{0})\delta.\]
\end{proof}
\subsection{A perturbation theorem for strong polynomial solutions}
\label{univalentthm}
In the former subsection, the solutions we considered are locally univalent in $\overline{D}$. However, the solutions which have physical meaning are required to be univalent in $\overline{D}$. The following Lemma~\ref{univalence} states that these locally univalent solutions are univalent if they are close to a univalent solution.
\begin{lemma}
\label{univalence}
Given $g(\xi,t)\in C^{1}([0,T_{0}], H(\overline{D_{r}}))\cap O(\overline{D_{r}})$ and $1<r^{'}<r$, there exists $\eta(g, T_{0}, r^{'})>0$ such that if 
\[\max_{([0,T_{0}])}\left| f^{'}(\cdot,t)-g^{'}(\cdot,t)\right|_{M(r)}\leq\eta\]
where $f(\xi,t)\in C([0,T_{0}], H(D_{r})\cap C(\overline{D_{r}}))$, then for $0\leq t\leq T_{0}$,
\[f(\xi,t)\in O(\overline{D_{r'}}).\]
\end{lemma}
\begin{proof}
The proof is separated into two parts (a)-(b):\\
(a)First assume that 
\begin{equation}
\label{hotel}
\max_{([0,T_{0}])}\left| f^{'}(\cdot,t)-g^{'}(\cdot,t)\right|_{M(r)}\leq\frac{1}{2}\min_{(\overline{D_{r}},[0,T_{0}])}\left| g^{'}(z,t)\right|.
\end{equation}
We want to show that there exists $r_{0}>0$ such that for any fixed $z_{0}\in\overline{D_{r'}}$,
\[f(\cdot,t):\overline{D_{r_{0}}(z_{0})}\rightarrow f(\overline{D_{r_{0}}(z_{0})})\]
is univalent. It is sufficient to prove that 
\[Re\frac{f^{'}(z,t)(z-z_{0})}{f(z,t)-f(z_{0},t)}\geq\frac{1}{2}, \quad z\in D_{r_{0}}(z_{0})\]
which means the function is injective on $\partial D_{r_{0}}(z_{0})$ and therefore is injective for $z\in\overline{D_{r_{0}}(z_{0})}$.\par
Now fix $z_{0}\in\overline{D_{r'}}$. Since $f(z,t)$ is analytic in $D_{r}$,
\[f(z,t)=f(z_{0},t)+\sum_{n=1}^{\infty}\frac{f^{(n)}(z_{0},t)}{n!}(z-z_{0})^{n},\quad z\in D_{r} .\]
Let 
\[l=\min\{r^{'}, r-r^{'}\}, M=\frac{3}{2}\max_{(\overline{D_{r}},[0,T_{0}])}\left| g^{'}\right| , m=\frac{1}{2}\min_{(\overline{D_{r}},[0,T_{0}])}\left| g^{'}\right|.\]
By (\ref{hotel}),  we can get that 
\[\max_{(\overline{D_{r}},[0,T_{0}])}\left| f(z,t)\right|\leq M, \quad\min_{(\overline{D_{r}},[0,T_{0}])}\left| f^{'}(z,t)\right|\geq m.\]
Note that
\[\left|\frac{f^{(n)}(z_{0},t)}{n!}\right|\leq Ml^{-(n)}, \quad n\geq 1.\]
Pick $0<r_{0}< l$ such that $\sum_{n=2}^{\infty}Ml^{-n}r_{0}^{n-1}(n-1)\leq\frac{m}{4}$. For $|z-z_{0}|<r_{0}$, we have
\begin{align}
&\left|\frac{f^{'}(z,t)(z-z_{0})}{f(z,t)-f(z_{0},t)}-1\right|\notag\\
=&\left|\frac{\sum_{n=1}^{\infty}\frac{f^{(n)}(z_{0},t)}{n!}(z-z_{0})^{n-1}n}{\sum_{n=1}^{\infty}\frac{f^{(n)}(z_{0},t)}{n!}(z-z_{0})^{n-1}}-1\right|\notag\\
=&\left|\frac{\sum_{n=2}^{\infty}\frac{f^{(n)}(z_{0},t)}{n!}(z-z_{0})^{n-1}(n-1)}{f^{'}(z_{0},t)+\sum_{n=2}^{\infty}\frac{f^{(n)}(z_{0},t)}{n!}(z-z_{0})^{n-1}}\right|\notag\\
\leq&\frac{\sum_{n=2}^{\infty}Ml^{-n}\left| z-z_{0}\right|^{n-1}(n-1)}{m-\sum_{n=2}^{\infty}Ml^{-n}\left| z-z_{0}\right|^{n-1}}\leq\frac{1}{2}.\notag
\end{align}
It follows from the above inequality that
\[Re\frac{f^{'}(z,t)(z-z_{0})}{f(z,t)-f(z_{0},t)}\geq\frac{1}{2}, \quad z\in D_{r_{0}}(z_{0}).\]
(b)Assume that there doesn't exist such $\eta>0$ such that the Lemma holds, then there exist $\eta_{k}$, $f^{k}(\xi,t)\in C^{1}([0,T_{0}], H(D_{r})\cap C(\overline{D_{r}}))$ and  $\xi_{k}^{1}, \xi_{k}^{2}\in \overline{D_{r'}}$ where $\xi_{k}^{1}\neq \xi_{k}^{2}$, such that\\
(1) $\eta_{k}$ goes to zero as $k$ goes to $\infty$;\\
(2) $f^{k}(\xi_{k}^{1},t_{k})=f^{k}(\xi_{k}^{2},t_{k})$;\\
(3)$\left| f^{k}(\xi_{k}^{1},t_{k})-g(\xi_{k}^{1},t_{k})\right|\leq\eta_{k},\left| f^{k}(\xi_{k}^{2},t_{k})-g(\xi_{k}^{2},t_{k})\right|\leq\eta_{k}.$\\
Without loss of generality, assume $t_{k}$ converges to $t_{0}$, $\xi_{k}^{1}$ converges to $\xi^{1}$ and $\xi_{k}^{2}$ converges to $\xi^{2}$. Note that $\mid \xi^{1}-\xi^{2}\mid\geq r_{0}$. This implies
\[g(\xi^{1},t_{0})=g(\xi^{2},t_{0}).\]
This contradicts the assumption that $g(\xi,t_{0})$ is univalent in $\overline{D_{r}}$. Therefore, there exists $\eta>0$ such that the Lemma holds.
\end{proof}
\begin{proof}{\textbf{(proof of Theorem~\ref{twins})}}\\
(a)By Lemma~\ref{univalence}, there exists $\eta(f_{k_{0}}, T_{0}, r^{'})>0$ such that if
$f(\xi,t)$ satisfies
\[f(\xi,t)\in C^{1}([0,T_{0}],H(D_{r}))\mbox{\quad and\quad}\max_{([0,T_{0}])}\left| f_{k_{0}}'(\cdot,t)-f'(\cdot,t)\right|_{M(r)}\leq\eta,\]
then $f(\xi,t)\in O(\overline{D_{r'}})$ for $t\in [0, T_{0}]$.\\
(b)We apply Theorem~\ref{Main Lemma} by letting $t_{1}=T_{0}$, $l=\frac{1}{2}$ , $\delta$ small enough such that
\[\delta<\min_{1\leq j\leq k}\left\{\frac{\epsilon}{c(j,k_{0})}\right\}, \quad\delta<\frac{l}{\sqrt{k_{0}}}\min_{(\overline{D_{r}},[0,T_{0}])}\left| f_{k_{0}}^{'}(\xi,t)\right|,\quad\delta<\min_{1\leq j\leq k}\left\{\frac{\eta}{c(j,k_{0})}\right\}\]
 and $\rho>1$ large enough such that $\frac{1}{Ck_{0}}(\ln\rho-\ln r)\geq T_{0}$. We get that  for $0\leq n\leq k$, $0\leq t\leq T_{0},$ the strong* solution to  (\ref{PG2}) $f(\xi,t)$ satisfies
\[\left| f_{k_{0}}^{(n)}(\cdot,t)-f^{(n)}(\cdot,t)\right|_{M(r)}<\min\{\epsilon,\eta\}.\]
Therefore $f(\xi,t)\in O(\overline{D_{r'}})$ and hence is a strong solution to  (\ref{PG1}).
\end{proof}

\section{Application-Evolution of perturbed disks in the suction case}
\label{sec9}
In this section, we aim to characterize the evolution of perturbed disks in the suction case.
\begin{lemma}
\label{suction}
Given $f_{k_{0}}(\xi,0)\in O(\overline{D})$ which is a polynomial of degree $k_{0}$. Let $f_{k_{0}}(\xi,t)$ be the strong solution to (\ref{PGS}) and the strong solution cease to exist as $t=b$. Then given $0<T_{0}<b$, there exist $\rho>1$ and $\delta>0$ such that, if $\|f(\xi,0)-f_{k_{0}}(\xi,0)\|_{\rho,1}<\delta$, then the solution $f(\xi,t)$ to (\ref{PGS}) exists for $0\leq t\leq T_{0}$.
\end{lemma}
\begin{proof}
(a)There exists $r>1$ such that $f_{k_{0}}(\xi,t)\in O(\overline{D_{r}})$ for all $0<t<T_{0}$.\\
(b)By Theorem~\ref{twins}, we are done with the proof.
\end{proof}
\begin{thm}
If the initial domain is close to a disk, then most of fluid is sucked before the corresponding strong solution to (\ref{PGS}) blows up.
\end{thm}
\begin{proof}
Assume that the disk is with area $\pi$ and therefore the conformal mapping is $f_{1}(\xi,0)=\xi$. The strong solution to (\ref{PGS}) is $f_{1}(\xi,t)=\sqrt{1-2t}\xi$ and the fluid is sucked out as $t=b=\frac{1}{2}$.\par
For $T_{0}<b$, we apply Lemma~\ref{suction} and obtain that  there exist $\rho>1$ and $\delta>0$ such that, if $\|f(\xi,0)-f_{1}(\xi,0)\|_{\rho,1}<\delta$, then the solution $f(\xi,t)$ to (\ref{PGS})  exists for $0<t<T_{0}$. If $b-T_{0}$ is small, the results show that most of fluid will be sucked before the strong solution $f(\xi,t)$ blows up.
\end{proof}

\section{Application-Large-time rescaling behaviors for large data and moments in the injection case}
\label{sec5}
In Richardson~\cite{richardson}, given $\Omega(t)$ which solves the Hele-Shaw  problem  with injection, the Richardson complex moments $\{M_{k}(t)\}_{k\geq 0}$ are defined by
\[M_{k}(t)=\frac{1}{\pi}\int_{\Omega(t)}z^{k}dxdy,\quad z=x+iy.\]
The quantity $M_{0}(t)\pi=\sqrt{2t+M_{0}(0)}\pi$ is the area of $\Omega(t)$ and $M_{k}(t), k\geq 1$ are conserved. Denote $\Omega'(t)=\{\frac{x}{\sqrt{2t+M_{0}(0)}}\mid x\in\Omega(t)\}$ which has area $\pi$ always. \par
Recall the definition of a strongly starlike function as in Gustafsson, Prokhorov and Vasil'ev~\cite{gustaf2} and Pommerenke~\cite{pomm}. A function $f\in O(D)$ is said to be  strongly starlike if there exists $\alpha\in (0,1]$ such that 
\[\left|\arg\frac{\xi f^{'}(\xi)}{f(\xi)}\right|<\alpha\frac{\pi}{2}, \quad\xi\in D.\]
Such a  function is also called a strongly starlike function of order $\alpha$.\par
In the case that  $\Omega(t)=f(D,t)$ where $f(\xi,t)$ is a global strong solution and is strongly starlike for $t\geq T_{0}$, $\partial\Omega'(t), t\geq T_{0}$ can be expressed by  a polar coordinate equation $(1+\overline{r}_{f}(t,\theta),\theta)$ for some $\overline{r}_{f}(t,\cdot):S^{1}\rightarrow [-1,\infty)$. The function $\overline{r}_{f}(t,\theta)$ satisfies
\[\overline{r}_{f}(t,\theta)=\frac{\left| f(\xi,t)\right|}{\sqrt{2t+M_{0}(0)}}-1,\quad t\geq T_{0}\]
where $\theta=\arg\frac{f(\xi,t)}{\mid f(\xi,t)\mid}$ for $\xi$ on $\partial D$. The value $\overline{r}_{f}(t,\theta)$ is well-defined if the function $f(\xi,t)$ is strongly starlike.\par
Define $M_{k}(f), k\geq 1$ to be the moments corresponding to the moving domain $\Omega(t)=f(D,t)$ where  $f(\xi,t)$ is a strong solution to (\ref{PG1}). In this section, we aim to prove Theorem~\ref{thm3.9} as follows:
\begin{thm}\label{thm3.9}
Given a global strong degree $k_{0}$ polynomial solution to $(\ref{PG1})$ $\{f_{k_{0}}(\xi,t)\}_{t\geq 0}$.\\
(a)There exist $\rho(f_{k_{0}})>1, \delta(f_{k_{0}})>0, T_{0}(f_{k_{0}})>0$ such that if $\|f(\cdot,0)-f_{k_{0}}(\cdot,0)\|_{\rho,3}<\delta$, then the strong solution to $(\ref{PG1})$ $f(\xi,t)$  is global and is a family of strongly starlike functions of order $<1$ for $t\geq T_{0}$.\\
(b)If $n_{0}=\min\{k\geq 1|M_{k}(f)\neq 0\}$, then
\[\lim_{T_{0}\leq t\rightarrow\infty}\|\overline{r}_{f}(t,\cdot)\|_{C^{2,\alpha}(S^{1})}(t)^{\lambda}=0,\quad \forall\lambda\in \left(0,1+\frac{n_{0}}{2}\right),\]
where $\overline{r}_{f}(t,\theta)=\frac{\mid f(\xi,t)\mid}{\sqrt{2t+M_{0}(0)}}-1$ and $\theta=\arg f(\xi,t)$, which are well-defined for $t\geq T_{0}$.
\end{thm} \par
The proof of Theorem~\ref{thm3.9} is given in subsection~\ref{proof1}. A geometric characterization of results in Theorem~\ref{thm3.9} is given in subsection \ref{characterization1}.
\subsection{Proofs for Theorem~\ref{thm3.9} }
\label{proof1}
We start with lemmas before the proof of Theorem~\ref{thm3.9}.
\begin{lemma}
\label{cor3.6}
Given a global strong solution $f(\xi,t)$ which is strongly starlike of order $<1$. There exists $\delta'>0$, such that if $\|\overline{r}_{f}(0,\cdot)\|_{C^{2,\alpha}(S^{1})}<\delta'$, then 
\[\limsup_{t\rightarrow\infty}\|\overline{r}_{f}(t,\cdot)\|_{C^{2,\alpha}(S^{1})}(2t)^{\lambda}=0,\quad\forall\lambda\in \left(0,1+\frac{n_{0}}{2}\right)\]
where $n_{0}=\min\{k\geq 1\mid M_{k}(f)\not=0\}$.
\end{lemma}
\begin{proof}
Let $g(\xi,\tau)=\frac{f(\xi,t)}{\sqrt{M_{0}(0)}}$ where $\tau=\frac{2\pi t}{M_{0}(0)}$. Then
\[Re\left[g_{\tau}\overline{g^{'}\xi}\right]=\frac{1}{2\pi}, \xi\in D\mbox{\quad and}\quad \left|g(D,0)\right|=\pi.\]
Since the boundary of $g(D,\tau)$ is analytic, then $\overline{r}_{g}(t,\cdot)\in h^{2,\alpha}(S^{1})$ where $h^{2,\alpha}(S^{1})$ is the little H$\ddot{o}$lder space as defined in Vondenhoff~\cite{vondenhoff}. Then by Theorem 3.3 and Theorem 4.3 in Vondenhoff~\cite{vondenhoff}, we obtain that there exists $\delta'>0$ such that if $\|\overline{r}_{g}(0,\cdot)\|_{C^{2,\alpha}(S^{1})}<\delta'$, then
\[\limsup_{\tau\rightarrow\infty}\|\overline{r}_{g}(\tau,\cdot)\|_{C^{2,\alpha}(S^{1})}(2\tau)^{\lambda}=0,\quad\forall\lambda\in \left(0,1+\frac{n_{0}}{2}\right)\]
where $n_{0}=\min\{k\geq 1\mid M_{k}(f)\not=0\}=\min\{k\geq 1\mid M_{k}(g)\not=0\}$. Here $\|\overline{r}_{f}(t,\cdot)\|_{C^{2,\alpha}(S^{1})}=\|\overline{r}_{g}(\tau,\cdot)\|_{C^{2,\alpha}(S^{1})}$. Therefore, we conclude that if $\|\overline{r}_{f}(0,\cdot)\|_{C^{2,\alpha}(S^{1})}<\delta'$, 
\[\limsup_{t\rightarrow\infty}\|\overline{r}_{f}(t,\cdot)\|_{C^{2,\alpha}(S^{1})}(2t)^{\lambda}=0,\quad\forall\lambda\in \left(0,1+\frac{n_{0}}{2}\right)\]
where $n_{0}=\min\{k\geq 1\mid M_{k}(f)\not=0\}$.
\end{proof}
\begin{lemma}\label{Theorem2.6}
Given a global strong degree $k_{0}$ polynomial solution $f_{k_{0}}(\xi,t)$ to $(\ref{PG1})$, then there exists $r>1$ such that for $t\geq 0$,
\[f_{k_{0}}(\xi,t)\in O(\overline{D_{r}}).\]
Also given $\epsilon>0, T_{0}>0, k\in N$ and $1<r^{'}<r$, there exist $\delta(f_{k_{0}},T_{0},\epsilon,k, r')>0$ and $\rho(f_{k_{0}},T_{0},\epsilon,k,r')>1$ such that if 
$\|f(\cdot,0)-f_{k_{0}}(\cdot,0)\|_{\rho,k}<\delta$ where $f(0,0)=0$ and $f^{'}(0,0)>0$,  then the strong solution $f(\xi,t)$ to (\ref{PG1}) satisfies
\[f(\xi,t)\in O(\overline{D_{r^{'}}})\cap C^{1}([0,T_{0}],H(D_{r})),\]
and for $0\leq n\leq k$, $0\leq t\leq T_{0},$
\[\left| f_{k_{0}}^{(n)}(\cdot,t)-f^{(n)}(\cdot,t)\right|_{M(r)}<\epsilon.\]
\end{lemma}
\begin{proof}
(a)There exists $r>1$ such that $f_{k_{0}}(\xi,t)\in O(\overline{D_{r}})$ for all $t>0$.\\
(b)By Theorem~\ref{twins}, we are done with the proof. \end{proof}
\begin{lemma}\label{Lemma3.7}
Define $M_{0}\pi$ as the area of $f(D)$ for some $f(\xi)=\sum_{i=1}^{\infty}a_{i}\xi^{i}$ in $O(\overline{D})$. Given $\delta'>0$, there exists $\epsilon'>0$ such that if $\mid\frac{f^{(j)}}{a_{1}}\mid_{M}<\epsilon'$  for $2\leq j\leq 3$, then $f(\xi)$ is strongly starlike of order $<1$ and $\|\overline{r}_{f}\|_{C^{2,\alpha}(S^{1})}<\delta'$ where $\overline{r}_{f}(\theta)=\frac{\mid f(\xi)\mid}{\sqrt{M_{0}}}-1$ and $\theta=\arg f(\xi)$.\end{lemma}
\begin{proof}
If $\epsilon'<1$, then $|\frac{f''}{a_{1}}|_{M}<1$. This implies that $\sum_{n=2}^{\infty}n|a_{n}|<|a_{1}|$ which is a sufficient condition for coefficients of strongly starlike functions; see Pommerenke~\cite{pomm}.\par
Now we treat the quantity $\|\overline{r}_{f}\|_{C^{2,\alpha}(S^{1})}$ by calculating $\max_{\theta\in S^{1}}|\partial_{\theta}^{(j)}\overline{r}_{f}|, 0\leq j\leq 3$. Note that  the value $M_{0}$ can be represented by $a_{1}^2+\sum_{n=2}^{\infty}n|a_{n}|^2$. The function $\overline{r}_{f}$ satisfies
\[\max_{\theta\in S^{1}}|\overline{r}_{f}|\leq\left|\frac{a_{1}}{\sqrt{M_{0}}}-1\right|+\sum_{n=2}^{\infty}\left|\frac{a_{n}}{\sqrt{M_{0}}}\right|\]
which goes to zero as $\epsilon'$ goes to zero. The function $\partial_{\theta}\overline{r}_{f}$ satisfies
\[\max_{\theta\in S^{1}}|\partial_{\theta}\overline{r}_{f}|=\max_{\xi\in\partial D}\left|\frac{1}{Re\left[\frac{f^{'}\xi}{f}\right]}\frac{Im\left[\xi f^{'}\overline{f}\right]}{|f|\sqrt{M_{0}}}\right|\]
which goes to zero as $\epsilon'$ goes to zero. Similarly, $\max_{\theta\in S^{1}}|\partial_{\theta}^2 \overline{r}_{f}|$ and $\max_{\theta\in S^{1}}|\partial_{\theta}^3 \overline{r}_{f}|$ go to zero as $\epsilon'$ goes to zero. We conclude that $\|\overline{r}_{f}\|_{C^{2,\alpha}(S^{1})}$ goes to zero as $\epsilon'$ goes to zero.\par
Finally, there exists $0<\epsilon'<1$ such that the theorem holds. \par
\end{proof}
\begin{proof}{\textbf{(proof of Theorem~\ref{thm3.9})}}\\
(a)Denote 
\[f(\xi,t)=\sum_{i=1}^{\infty}b_{i}(t)\xi^{i};\quad f_{k_{0}}(\xi,t)=\sum_{i=1}^{k_{0}}a_{i}(t)\xi^{i}.\]
 Note that $b_{1}^2(t)\geq b_{1}^2(0)+2t$ and $a_{1}^2(t)\geq a_{1}^2(0)+2t$ as shown in~Kuznetsova~\cite{kuz}. We separate the proof for (a) into (1)-(5) as follows:\\
(1)There exists $\delta'>0$ as stated in Lemma~\ref{cor3.6}.\\
(2)For such $\delta'>0$, we can find $\epsilon'>0$ as stated in Lemma~\ref{Lemma3.7}.\\
(3)Given $\epsilon'>0$, there exists $T_{0}>\frac{1}{2}$ such that for $t\geq T_{0}$,
\begin{equation}
\label{assumption2}
\left| \frac{f_{k_{0}}^{(2)}(\cdot,t)}{a_{1}(t)}\right|_{M}<\frac{1}{8}\epsilon'\mbox{\quad and \quad}\left| \frac{f_{k_{0}}^{(3)}(\cdot,t)}{a_{1}(t)}\right|_{M}<\frac{1}{8}\epsilon'
\end{equation}
since the coefficients $\{a_{i}(t)\}_{i\geq 2}$ are bounded and $a_{1}(t)\geq\sqrt{2t+a_{1}^2(0)}$ as shown in Kuznetsova~\cite{kuz}.\\
(4)By Lemma~\ref{Theorem2.6}, for such $T_{0}$ and $\epsilon'$, there exist $\rho>1$ and $\delta>0$ such that if  $\|f(\cdot,0)-f_{k_{0}}(\cdot,0)\|_{\rho,3}<\delta$, then\\
(i)the strong solution to (\ref{PG1}) $f(\xi,t)$ exists for $t\in [0,T_{0}]$, and\\
(ii)for $0\leq t\leq T_{0},$ $1\leq j\leq 3$, 
\begin{equation}
\label{who}
\left| f_{k_{0}}^{(j)}(\cdot,t)-f^{(j)}(\cdot,t)\right|_{M}<\min\left\{\frac{1}{2}a_{1}(T_{0}), \frac{1}{8}\epsilon'\right\}.
\end{equation}
From (\ref{who}) and the fact that $T_{0}\geq 1$, we also can obtain that $b_{1}(T_{0})\geq\max\{1, \frac{1}{2}a_{1}(T_{0})\}$. Therefore, by $(\ref{assumption2})$, (\ref{who}) and the fact that $b_{1}(T_{0})\geq\max\{1, \frac{1}{2}a_{1}(T_{0})\}$, we have
\[\left|\frac{f^{(j)}(\cdot,T_{0})}{b_{1}(T_{0})}\right|_{M}\leq \frac{1}{2}\epsilon', \quad 2\leq j\leq 3.\]
Due to the fact in (2), $f(\xi,T_{0})$ is strongly starlike of order $<1$ and
\[\|\overline{r}_{f}(T_{0},\cdot)\|_{C^{2,\alpha}(S^{1})}<\delta',\]
where $\overline{r}_{f}(t,\theta)=\frac{\mid f(\xi,t)\mid}{\sqrt{M_{0}(t)}}-1$ and $\theta=\arg f(\xi, t)$.\\
(5)By (1)-(4), we conclude that there exist $T_{0}>0$, $\rho>1$, $\delta>0$ such that  if $\|f(\cdot,0)-f_{k_{0}}(\cdot,0)\|_{\rho,3}<\delta$, then\\
(i)the strong solution $f(\xi,t)$ exists for $t\in [0,T_{0}]$, and\\
(ii)$f(\xi,T_{0})\in O(\overline{D})$ is a strongly starlike function of order $<1$, and\\
(iii)$\|\overline{r}_{f}(T_{0},\cdot)\|_{C^{2,\alpha}(S^{1})}<\delta'.$\\
By Theorem 2.1 in Gustafsson, Prokhorov and Vasil'ev~\cite{gustaf2}, the solution $f(\xi,t)$ must be global and $f(\xi,t), t\geq T_{0}$ has  strictly decreasing strongly starlike order $\alpha(t)$ since $f(\xi,T_{0})\in O(\overline{D})$ and is a strongly starlike function. This also implies that $\overline{r}_{f}(t,\cdot)$ is well-defined for $t\geq T_{0}$.\\
(b)From (5), the assumptions in Lemma~\ref{cor3.6} are satisfied and we obtain 
\[\limsup_{T_{0}\leq t\rightarrow\infty}\|\overline{r}_{f}(t,\cdot)\|_{C^{2,\alpha}(S^{1})}(2t)^{\lambda}=0, \quad\forall\lambda\in \left(0,1+\frac{n_{0}}{2}\right).\]
\end{proof}
\subsection{Geometric meaning of rescaling behavior in Theorem~\ref{thm3.9}}
\label{characterization1}
The initial domains we consider  in this section are
\[\{f_{k_{0}}(D,0)\mid  \mbox{$f_{k_{0}}(\xi,t)$ is a global strong polynomial solution of degree $k_{0}\in N$}\}\]
and small perturbations of them. Theorem~\ref{thm3.9} demonstrates that starting with an initial domain $\Omega(0)$ as above,  we can obtain a global solution $\Omega(t)$ which is simply connected and has a real analytic boundary, and a rescaling behavior is given in terms of moments. Here we aim to give a geometric characterization for  this rescaling behavior by carrying out some explicit calculation:
\begin{thm}
\label{geometric}
Given a global strong solution $f(\xi,t)$ where $f(\xi,0)$ satisfies the assumption of Theorem~\ref{thm3.9} and $\Omega(t)=f(D,t)$.  We show that the rescaled domain $\Omega'(t)=\{x|x\sqrt{|\Omega(t)|/\pi}\in\Omega(t)\}$ has  radius satisfy that
\[\max_{z\in\partial\Omega^{'}(t)}\left|\left| z\right|-1\right|=o\left(\frac{1}{t}\right)^{\lambda},\quad\forall \lambda\in \left(0,1+\frac{n_{0}}{2}\right)\]
and its curvature $\kappa(t,z), z\in\Omega'(t)$ satisfies
\[\max_{z\in\Omega^{'}(t)}\left|\kappa(t,z)-1\right|=o\left(\frac{1}{t}\right)^{\lambda},\quad\forall \lambda\in \left(0,1+\frac{n_{0}}{2}\right),\]
where $n_{0}=\min\{k\geq 1|M_{k}(f)\neq 0\}$.
\end{thm}
\begin{proof}
 Let $f(\xi,t)$ be a global strong solution satisfies Theorem~\ref{thm3.9}. There exists $T_{0}>0$ such that $\overline{r}_{f}(t,\theta), t\geq T_{0}$ is well-defined. The value $\left|\kappa(t,z)-1\right|$ satisfies
\begin{equation}
\label{curvature1}
|\kappa-1|=\left|\frac{\left(1+\overline{r}_{f}\right)^2+2\left(\overline{r}_{f}^{'}\right)^2-\overline{r}_{f}^{''}\left(1+\overline{r}_{f}\right)}{\left[\left(1+\overline{r}_{f}\right)^2+\left(\overline{r}_{f}^{'}\right)^2\right]^{\frac{3}{2}}}-1\right|=O(\|\overline{r}_{f}\|_{C^{2}(S^{1})})
\end{equation} 
as $\|\overline{r}_{f}\|_{C^{2}}$ approaches $0$. Since $\|\overline{r}_{f}\|_{C^{2,\alpha}(S^{1})}=o(\frac{1}{t})^{\lambda}, \forall \lambda\in (0,1+\frac{n_{0}}{2})$
 by the results in Theorem~\ref{thm3.9},  we can obtain from (\ref{curvature1}) that 
 \[\max_{z\in\Omega^{'}(t)}|\kappa(t,z)-1|=o\left(\frac{1}{t}\right)^{\lambda},\quad\forall \lambda\in \left(0,1+\frac{n_{0}}{2}\right).\]
 Similarly, since $\|\overline{r}_{f}\|_{C^{2,\alpha}(S^{1})}=o(\frac{1}{t})^{\lambda}, \forall \lambda\in (0,1+\frac{n_{0}}{2})$ by the results in Theorem~\ref{thm3.9},  we can obtain that radius satisfies
\[\max_{z\in\partial\Omega^{'}(t)}|\left| z\right|-1|=o\left(\frac{1}{t}\right)^{\lambda},\quad\forall \lambda\in \left(0,1+\frac{n_{0}}{2}\right).\]
\end{proof}
\section{Existence and uniqueness proof of the P-G equation}
\label{sec6}
In this section, we assume the short-time well-posedness of strong* polynomial solutions as shown in Gustfasson~\cite{gustaf1} and we give a shorter proof of short-time well-posedness  for strong solutions in the injection case. Especially, the proof of short-time existence  of strong solutions is  an application of Theorem~\ref{Main Lemma}  and this proof implies that every strong solution can be approximated by many strong* polynomial solutions locally in time. The uniqueness proof is given separately.
\subsection{Existence}
\begin{thm}
\label{thm2.6}
Given $f(\xi,0)\in \omega(\overline{D_{r}})\cap H(\overline{D_{\rho_{0}}})$ where $\rho_{0}>r>1$, then there exist $t_{0}>0$ and a strong* solution  to (\ref{PG2}) $f(\xi,t)\in C^{1}([0,t_{0}],H(D_{r}))\cap \omega(D_{r})$ with the intial value $f(\xi,0)$.
\end{thm}
\begin{proof}
(a).For $f(\xi,0)=\sum_{i=1}^{\infty}a_{i}(0)\xi^{i}\in H(\overline{D_{\rho_{0}}})$, there exists $M>0$ such that
\[\mid a_{i}(0)\mid\leq M\rho_{0}^{-i}.\]
Define $f_{n}(\xi,0)=\sum_{i=1}^{n}a_{i}(0)\xi^{i}.$
Then
\[\left|\min_{\overline{D_{r}}}\left| f^{'}(\cdot,0)\right|-\min_{\overline{D_{r}}}\left| f_{n}^{'}(\cdot,0)\right|\right|\leq\sum_{i=n+1}^{\infty}i\mid a_{i}(0)\mid(r)^{i}\leq\sum_{i=n+1}^{\infty}iM(\frac{\rho_{0}}{r})^{-i}\]
where $\sum_{i=n+1}^{\infty}iM(\frac{\rho_{0}}{r})^{-i}$ approaches zero as $n$ approaches $\infty$. Therefore
there exists $n_{0}\in N$ such that 
\[\frac{1}{2}\min_{\overline{D_{r}}}\left| f^{'}(\cdot,0)\right|\leq\min_{\overline{D_{r}}}\left| f_{n}^{'}(\cdot,0)\right|, \quad n\geq n_{0}\]
and $f_{n}(\xi,0)\in \omega(\overline{D_{r}})$. By Gustafsson~\cite{gustaf1}, there exists a strong* polynomial solution $f_{n}(\xi,t)\in\omega(\overline{D_{r}})$ at least for a short time.\\
(b).Given $1<r_{0}< \frac{\rho_{0}}{r}$, there exists $k_{0}\geq n_{0}$ such that 
\begin{equation}
\label{key1}
\sum_{k=k_{0}+1}^{\infty}|a_{k}(0)|\left(\frac{\rho_{0}}{r_{0}}\right)^{k}k^{3/2}\leq \frac{1}{\sqrt{k_{0}}}\frac{1}{8}\min_{\overline{D_{r}}}\left| f^{'}(\cdot,0)\right|
\end{equation}
\\
(c).There exists $t_{1}>0$ such that the strong* solution to (\ref{PG2}) $f_{k_{0}}(\xi,t)$ exists for $0\leq t\leq t_{1}$ and
\[\min_{(\overline{D_{r}},[0,t_{1}])}\left| f_{k_{0}}^{'}\right|\geq\frac{1}{4}\min_{\overline{D_{r}}}\left| f^{'}(\cdot,0)\right|.\]
By the above, $(\ref{key1})$ implies
\begin{equation}
\label{y4}
\sum_{k=k_{0}+1}^{\infty}\left| a_{k}(0)\right| \left(\frac{\rho_{0}}{r_{0}}\right)^{k}k^{3/2}\leq\frac{1}{\sqrt{k_{0}}}\frac{1}{2}\min_{(\overline{D_{r}},[0,t_{1}])}\left| f_{k_{0}}^{'}\right|.
\end{equation}
The inequality (\ref{y4}) implies that
\[\left\|f(\cdot,0)-f_{k_{0}}(\cdot,0)\right\|_{\frac{\rho_{0}}{r_{0}},1}\leq\frac{1}{\sqrt{k_{0}}}\frac{1}{2}\min_{(\overline{D_{r}},[0,t_{1}])}\left| f_{k_{0}}^{'}\right|.\]
(d).By letting $\rho=\frac{\rho_{0}}{r_{0}}$ and $l=\frac{1}{2}$, we can see that assumption (A) in Theorem~\ref{Main Lemma} is satisfied from (c). By applying Theorem~\ref{Main Lemma}, the short-time existence is proven.
\end{proof}
\begin{rem}The proof also can be applied to the suction case.\end{rem}
If we assume $f(\xi,0)$ is univalent, then $f(\xi,t)$ we obtained in Theorem~\ref{thm2.6} is also univalent in short time. Therefore, we obtain the following results.
\begin{thm}
\label{thm2.7}
Given $f(\xi,0)\in O(\overline{D_{r}})\cap H(\overline{D_{\rho_{0}}})$ where $\rho_{0}>r>1$, then for $1<r'<r$, there exists $b>0$ and a strong  solution  to (\ref{PG1}) $f(\xi,t)\in C^{1}([0,b],H(\overline{D_{r'}}))\cap O(\overline{D_{r'}})$ with the intial value $f(\xi,0)$.
\end{thm}
Since for a given $f(\xi,0)\in O(\overline{D})$, there exist $1<r<\rho_{0}$ such that $f(\xi,0)\in H(\overline{D_{\rho_{0}}})\cap O(\overline{D_{r}})$, Theorem~\ref{thm2.7} implies the following directly:
\begin{thm}
Given $f(\xi,0)\in O(\overline{D})$, there exists a strong solution to (\ref{PG1}) $f(\xi,t)$ locally in time. 
\end{thm}
\subsection{Uniqueness}
\begin{thm}\label{uniqueness} Strong solutions to (\ref{PG1}) are unique.\end{thm}
\begin{proof}
(1) Let $f(\xi,0)\in O(\overline{D})$. Assume there are two strong solutions $f_{1}, f_{2}$  with the same initial value $f(\xi,0)$. There exist $1<r'$ and $b>0$ such that $f_{i}(\xi,t)\in O(\overline{D_{r'}})$  for $0\leq t\leq b$ and $f_{i}(\xi,t)$ is continuous in $(\overline{D_{r'}},[0,b])$ for $1\leq i\leq 2$. Denote 
\[M^2=\max_{i=1,2}\max_{t\in[0,b]}\int_{\partial D_{r'}}\left| f_{i}^{'}\right|^2d\theta\]then
\[\left| \alpha_{i}(t)\right|\leq \frac{M}{i}(r')^{-i},\quad\left| \beta_{i}(t)\right|\leq \frac{M}{i}\left(r'\right)^{-i}\]
if we denote $f_{1}(\xi,t)=\sum_{i=1}^{\infty}\alpha_{i}(t)\xi^i$ and $f_{2}(\xi,t)=\sum_{i=1}^{\infty}\beta_{i}(t)\xi^i$.\\
(2)By (\ref{calp}), 
\begin{align}
\label{qq}
&\left\|\frac{d}{dt}[f_{1}-f_{2}]\right\|_{L^{2}([0,2\pi])}\notag\\
\leq&\left\{\max_{\partial D}\left|P\left[\frac{1}{\left| f_{2}^{'}\right|^2}\right]\right|+C_{2}\max_{\partial D}\left| f_{1}^{'}\right|\max_{\partial D}\frac{|f_{1}^{'}|+|f_{2}^{'}|}{|f_{1}^{'}|^2 |f_{2}^{'}|^2}\right\}\left\|f_{1}^{'}-f_{2}^{'}\right\|_{L^{2}([0,2\pi])}.
\end{align} 
Therefore, by (\ref{qq}), there exists $C>0$, for $t\in[0,b]$
\begin{align}       
&\sum_{i=1}^{\infty}\left[\left|(\alpha_{i}-\beta_{i})_{t}\right|\right]^2\notag\\
\leq& C\left\{\sum_{i=1}^{\infty}\left[\left|(\alpha_{i}-\beta_{i})\right| i\right]^2\right\}\notag\\
\leq &C\left\{\sum_{i=1}^{k}\left[\left|(\alpha_{i}-\beta_{i})\right| i\right]^2+\sum_{i=k+1}^{\infty}(2M)^2\left(r'\right)^{-2i}\right\}\notag\\
\leq &C\left\{\sum_{i=1}^{k}\left[\mid(\alpha_{i}-\beta_{i})\mid i\right]^2+4M^2\left(\frac{\left(r'\right)^{-2(k+1)}}{1-\left(r'\right)^{-2}}\right)\right\}\notag.
\end{align} 
(3)Denote  $D_{k}(t)=\sum_{i=1}^{k}[\mid(\alpha_{i}-\beta_{i})\mid i]^2$, then
\begin{align} 
 D_{k}^{'}(t)&=\sum_{i=1}^{k}2Re\left[(\alpha_{i}-\beta_{i})\overline{(\alpha_{i}-\beta_{i})_{t}}\right]i^2\notag\\
&\leq 2k\left\{\sum_{i=1}^{k}\left[\left|(\alpha_{i}-\beta_{i})\right| i\right]^2\right\}^{1/2}\left\{\sum_{i=1}^{k}\left[\left|(\alpha_{i}-\beta_{i})_{t}\right|\right]^2\right\}^{1/2}\notag\\
&\leq 2kCD_{k}^{1/2}(t)\left\{D_{k}(t)+4M^2\left(\frac{\left(r'\right)^{-2(k+1)}}{1-\left(r'\right)^{-2}}\right)\right\}^{1/2}\notag\\
&\leq 2kCD_{k}^{1/2}(t)\left\{D_{k}^{1/2}(t)+2M\left(\frac{\left(r'\right)^{-(k+1)}}{(1-\left(r'\right)^{-2})^{1/2}}\right)\right\}\notag\\
&\leq 2kCD_{k}(t)+4kMCD_{k}^{1/2}(t)\left(\frac{\left(r'\right)^{-(k+1)}}{(1-\left(r'\right)^{-2})^{1/2}}\right).\notag
\end{align}
Note that $|\Omega(t)|=\pi\sum_{i=1}^{\infty}i\mid \alpha_{i}(t)\mid^2=\pi\sum_{i=1}^{\infty}i\mid \beta_{i}(t)\mid^2\leq |\Omega(0)|+2\pi b$ where $|\Omega(t)|$ is the area of the moving domain at time $t$. So we have $D_{k}(t)\leq \frac{1}{\pi}4k|\Omega(t)|\leq \frac{1}{\pi}4k(|\Omega(0)|+2\pi b)=2kA$ for some $A>0$. Therefore
\[D_{k}^{'}(t)\leq 2kCD_{k}(t)+4MC(2A)^{1/2}k^{3/2}\frac{\left(r'\right)^{-(k+1)}}{(1-\left(r'\right)^{-2})^{1/2}}.\]
Denote $(2A)^{1/2}(4MC)\frac{1}{(1-\left(r'\right)^{-2})^{1/2}}=C_{0}$, then
\[D_{k}^{'}(t)\leq 2kCD_{k}(t)+C_{0}\left(r'\right)^{-(k+1)}k^{3/2}\]
\[\left(D_{k}(t)e^{-2kCt}\right)^{'}\leq e^{-2kCt}C_{0}\left(r'\right)^{-(k+1)}k^{3/2}\]
\[D_{k}(t)e^{-2Ckt}\leq \frac{1-e^{-2kCt}}{2kC}C_{0}\left(r'\right)^{-(k+1)}k^{3/2}\]
\begin{equation}
\label{fun}
D_{k}(t)\leq\frac{1}{2kC}\left(e^{2kCt}\right)C_{0}\left(r'\right)^{-(k+1)}k^{3/2}=\frac{1}{2r'C}\left(e^{2Ct}\left(r'\right)^{-1}\right)^{k}k^{\frac{1}{2}}C_{0}.
\end{equation}
For $0\leq t<\frac{1}{2C}\ln r'$, in (\ref{fun}) we let $k$ approach $\infty$, then $D_{k}(t)$ approaches zero since $\frac{1}{2C}(e^{2Ct}\left(r'\right)^{-1})^{k}k^{\frac{1}{2}}C_{0}$ approaches zero. Therefore $f_{1}(\xi,t)=f_{2}(\xi,t)$ for $t\in[0,T)$ where $T=\min\{\frac{1}{2C}\ln r',b\}$.\\
(4)Hence, the uniqueness of the short-time existence is proven.

\end{proof}

\pagebreak

\section*{Acknowledgements}
The author is indebted to her advisor, Govind Menon, for many things, including his constant guidance and important opinions. This material is based upon work supported by the National Science
Foundation under grant nos. DMS 06-05006 and DMS 07-48482.


\bibliography{main0}

\begin{thebibliography}{10}

\bibitem{gustaf1}
{\sc B.~Gustafsson}, {\em On a differential equation arising in a {H}ele-{S}haw
  flow moving boundary problem}, Ark. Mat., 22 (1984), pp.~251--268.

\bibitem{gustaf2}
{\sc B.~Gustafsson, D.~Prokhorov, and A.~Vasil'ev}, {\em Infinite lifetime for
  the starlike dynamics in {H}ele-{S}haw cells}, Proc. Amer. Math. Soc., 132
  (2004), pp.~2661--2669 (electronic).

\bibitem{sakai2}
{\sc B.~Gustafsson and M.~Sakai}, {\em On the curvature of the free boundary
  for the obstacle problem in two dimensions}, Monatsh. Math., 142 (2004),
  pp.~1--5.

\bibitem{kuz}
{\sc O.~S. Kuznetsova}, {\em On polynomial solutions of the {H}ele-{S}haw
  problem}, Sibirsk. Mat. Zh., 42 (2001), pp.~1084--1093, iii.

\bibitem{Lin}
{\sc Y.-L. Lin}, {\em Large-time rescaling behaviors of Stokes and Hele-Shaw
  flows driven by injection}, arXiv:0906.0916v1, (preprint).

\bibitem{pomm}
{\sc C.~Pommerenke}, {\em Univalent functions}, Vandenhoeck \& Ruprecht,
  G\"ottingen, 1975.
\newblock With a chapter on quadratic differentials by Gerd Jensen, Studia
  Mathematica/Mathematische Lehrb\"ucher, Band XXV.

\bibitem{reissig}
{\sc M.~Reissig and L.~von Wolfersdorf}, {\em A simplified proof for a moving
  boundary problem for {H}ele-{S}haw flows in the plane}, Ark. Mat., 31 (1993),
  pp.~101--116.

\bibitem{richardson}
{\sc S.~Richardson}, {\em Hele-{Shaw} flows with a free boundary produced by
  the injection of fluid into a narrow channel}, J. Fluid Mech., 56 (1972),
  pp.~609--618.

\bibitem{rudin}
{\sc W.~Rudin}, {\em Real and complex analysis}, McGraw-Hill Book Co., New
  York, third~ed., 1987.

\bibitem{sakai}
{\sc M.~Sakai}, {\em Sharp estimates of the distance from a fixed point to the
  frontier of a {H}ele-{S}haw flow}, Potential Anal., 8 (1998), pp.~277--302.

\bibitem{vondenhoff}
{\sc E.~Vondenhoff}, {\em Long-time asymptotics of {H}ele-{S}haw flow for
  perturbed balls with injection and suction}, Interfaces Free Bound., 10
  (2008), pp.~483--502.

\end{thebibliography}

\end{document}